\documentclass[onecolumn]{IEEEtran}
\usepackage{graphics,cite} 
\usepackage{amsmath, amssymb, fullpage}

\newtheorem{theorem}{Theorem}
\newtheorem{proposition}{Proposition}
\newtheorem{corollary}{Corollary}

\usepackage{standalone, amsmath, amssymb, graphicx, color}
\usepackage{url}
\usepackage{caption}
\usepackage[scriptsize,tight]{subfigure}

\newcommand{\vo}[1]{\boldsymbol{#1}} 
\newcommand{\mo}[1]{\boldsymbol{#1}} 

\newcommand{\real}{\mathbb{R}}

\newcommand{\oderiv}[2]{\frac{d #1}{d #2}}

\newcommand{\Exp}[1]{\boldsymbol{\mathsf{E}} \left[#1\right]}

\newcommand{\pdfp}{p(\param)}

\newcommand{\x}{\vo{x}} 
\newcommand{\xdot}{\dot{\vo{x}}} 
\renewcommand{\u}{\vo{u}} 
\newcommand{\param}{\vo{\Delta}} 
\newcommand{\domain}[1]{{\mathcal{D}_{#1}}} 
\newcommand{\Y}{\vo{Y}} 
\newcommand{\A}{\mo{A}} 
\let\B=\undefined
\newcommand{\B}{\mo{B}} 

\usepackage{xifthen}
\newcommand{\basis}[2]{%
 \phi_{#1}
  \ifthenelse{\isempty{#2}}%
    {}
    {({#2})}
}

\newcommand{\diag}{\boldsymbol{\mathsf{diag}}}

\newcommand{\xpc}{\x_{pc}} 
\newcommand{\xpcdot}{\dot{\x}_{pc}} 
\newcommand{\Apc}{\mo{A}_{pc}} 

\newcommand{\I}[1]{\vo{I}_{{#1}}} 

\renewcommand{\vec}[1]{\boldsymbol{\mathsf{vec}}\left({#1}\right)}

\newcommand{\X}{\mo{X}} 
 
\newcommand{\W}{\mo{W}} 
\newcommand{\K}{\mo{K}} 

\newcommand{\set}[1]{\mathcal{#1}}
\newcommand{\inner}[1]{\left\langle #1 \right\rangle}

\newcommand{\figlabel}[1]{\label{fig:#1}}
\newcommand{\eqnlabel}[1]{\label{eqn:#1}}

\newcommand{\eqn}[1]{\eqref{eqn:#1}}
\newcommand{\Eqn}[1]{Eqn.(\ref{eqn:#1})}
\newcommand{\fig}[1]{fig.(\ref{fig:#1})}
\newcommand{\Fig}[1]{Fig.(\ref{fig:#1})}

\newcommand{\etal}{\textit{et al. }}

\definecolor{darkgreen}{rgb}{0,0.65,0}

\renewcommand{\P}{\mo{P}} 

\newcommand{\Abar}{\mo{\bar{A}}}
\newcommand{\Bbar}{\mo{\bar{B}}}
\newcommand{\Qbar}{\mo{\bar{Q}}}
\newcommand{\Rbar}{\mo{\bar{R}}}

\date{} 

\title{Robust State Feedback Control Design with Probabilistic System Parameters}
\author{Raktim Bhattacharya}

\begin{document}
\maketitle
\thispagestyle{empty}
\pagestyle{empty}

\begin{abstract}
In this paper, a new polynomial chaos based framework for analyzing linear systems with probabilistic parameters is presented. Stability analysis and synthesis of optimal quadratically stabilizing controllers for such systems are presented as convex optimization problems, with  exponential mean square stability guarantees. A Monte-Carlo approach for analysis and synthesis is also presented, which is used to benchmark the polynomial chaos based approach. The computational advantage of the polynomial chaos approach is shown with an example based on the design of an optimal EMS-stabilizing controller, for an F-16 aircraft model.
\end{abstract}

\section{Introduction}
In this paper we study the problem of designing state feedback controllers for linear time invariant systems with probabilistic system parameters. Such systems in continuous time can be defined as
\begin{align}
 \dot{\x}	 &=  \A(\param)\x + \B(\param)\u, \eqnlabel{contiDyn}
\end{align}
where $\param\in\real^d$ is a vector of uncertain parameters, with joint probability density function $\pdfp$. Matrices $\A(\param)\in\real^{n\times n}$, $\B(\param) \in\real^{n\times m}$ are system matrices that depend on $\param$. Consequently, the solution $\x:=\x(t,\param)\in\real^n$ also depends on $\param$. The objective is to design a state-feedback law in the form $\u = \K\x$, which stabilizes the system in some suitable sense, where $\K\in\real^{m\times n}$. Thus, we are looking to obtain a constant deterministic gain $\K$ that stabilizes the systems in \eqn{contiDyn}, with probabilistic uncertainty in $\param$.

Stability of dynamical systems of the type 
\begin{align}
\dot{\x} = \A(\param)\x, \eqnlabel{simpleSys}
\end{align}
have been extensively studied in the framework of stochastic dynamical systems. Depending on the nature of $\param$, two approaches are commonly used. If $\param$ is Gaussian white noise, the solution process is a diffusion process and is analyzed using theory of Markov processes \cite{dynkin2012theory,bharucha2012elements} and Ito calculus \cite{itocalc}. If $\param$ is not Gaussian white noise, and has well defined samples properties, the system in \eqn{simpleSys} can be analyzed by ordinary rules of calculus. Such systems can be considered to be a collection of ordinary deterministic differential equations upon which a probability measure has been induced by the parameter process $\param$. The second class of systems are often what we actually encounter in engineering problems. The Gaussian white noise case is a mathematical abstraction and admits richer set of analysis tools. However, this should not be the motivation for assuming white noise process for systems with random parameters as it is a non-trivial matter and should be considered carefully \cite{eugene1965relation}. 

Early work on stability analysis of linear systems with randomly time-varying parameters was done by Rosenbloom \cite{rosenbloom1954analysis}, where he studied first order linear systems with stationary Gaussian parameter process. His work focussed on stability properties in terms of the moments. Due to the stationary properties of $\Delta$, the moments and hence the asymptotic properties of system could be studied explicitly. This work was extended by Bertram and Sarachik \cite{bertram1959stability} for linear diagonal systems and were first to provide a Lyapunov framework for studying such systems with random parameters. The parameters were also considered to be Gaussian and recovered Rosebloom's result, with a more general approach. They also considered the problem where $\param$ is random but piecewise constant. At the same time, Kats and Krosovskii \cite{kats1960stability} independently provided a Lyapunov framework for analyzing stability in probability and moments, for systems where the $\param$ process is a stationary Markov process with finite number of states. Bharucha \cite{bharucha1961stability} extended the work by Kats \etal by exploiting the fact that a stationary Markov parameter process admits a piecewise constant property in the linear system. Bharucha also showed that for such systems, asymptotic stability implies exponential stability. Palmer \cite{palmer1966sufficient} significantly extended the work by Kats \etal and Bharucha on linear systems with Markov coefficients by applying Markov chain theory and exploiting the induced piecewise constant property. The piecewise constant property was also exploited in the works of Morozan \cite{morozan1967stability} and Soeda \etal \cite{soeda1966stability}. The above described literature focussed on systems for which the solution can be obtained in closed form. More general linear systems of the form $\dot{\x} = (\A_0 + \A_1(\param))\x$, where $\A_0$ is a stability matrix and $\A_1(\param)$ is matrix whose non-zero coefficients are stationary, ergodic processes with almost surely continuous sample functions, were studied by Kozin \cite{kozin1963almost}. He provided sufficient conditions for asymptotic stability with probability one, for such systems. This work was refined further \cite{caughey1965almost,ariaratnam1967dynamic, morozan1967stability,wang1966almost, infante1968stability, gray1967frequency}. For related work when $\param$ is Gaussian white noise, please refer to \cite{khasminskii2011stochastic, kushner} and references therein.

In this paper, we focus on systems where $\param$ is a vector of random variables, which is a simpler problem than randomly time-varying parameters. The problem of analyzing systems with uncertain, but constant, parameters have also been addressed in the robust control literature. In that approach, the support of $\param$ is assumed to be polytopic, and stability is analyzed for parameter combinations along the vertices of the polytope \cite{bernussou1989linear,Boydetal:1994, boyd1989structured}. This is the so called ``worst-case'' approach for stability analysis. In this paper, we present new stability analysis and control design methods, which ensure exponential mean square stability for systems with probabilistic system parameters. These are developed in the polynomial chaos framework \cite{wiener}. A comparison with worst-case quadratic stability approach \cite{corless1994robust}, for systems with uniformly distributed $\param$ \cite{barmish3} is also performed, which highlights the conservativeness of the worst-case approach. 

The paper is organized as follows. First, we present a brief background on polynomial chaos theory and its application to linear systems with random parameters. This is followed by the main contribution of this paper, captured in propositions 1, 2, and 3; and theorems 1, 2, and 3. For benchmarking, we also present randomized algorithms for system analysis and design. We demonstrate the  computational efficiency of polynomial chaos framework with an example based on an F-16 aircraft model. The paper ends with a summary section.

\section{Polynomial Chaos Theory}
Polynomial chaos is a non-sampling based method to determine evolution of uncertainty in dynamical system with probabilistic system parameters \cite{pcFEM}. Very briefly, a general second order process $X(\omega)\in
\mathcal{L}_2(\Omega,\mathcal{F},P)$ can be expressed by polynomial
chaos as
\begin{equation}
\eqnlabel{gPC}
X(\omega) = \sum_{i=0}^{\infty} x_i\phi_i({\param}(\omega)),
\end{equation}
where $\omega$ is the random event and $\phi_i({\param}(\omega))$
denotes the polynomial chaos basis of degree $p$ in terms of the random variables
$\param(\omega)$. $(\Omega,\mathcal{F},P)$ is a probability space, where $\Omega$
is the sample space, $\mathcal{F}$ is the $\sigma$-algebra of the
subsets of $\Omega$, and $P$ is the probability measure. According to Cameron and Martin \cite{CameronMartin} such an expansion converges in the $\mathcal{L}_2$ sense for any arbitrary stochastic process with finite second moment. In practice, the infinite series is truncated and $X(\omega)$ is approximated by 
\[
X(\omega) \approx \hat{X}(\omega) = \sum_{i=0}^{N} x_i\phi_i({\param}(\omega)).
\] The functions $\{\phi_i\}$ are a family of
orthogonal basis in $\mathcal{L}_2(\Omega,\mathcal{F},P)$ satisfying
the relation
\begin{equation}
\Exp{\phi_i\phi_j}:= \int_{\mathcal{D}_{\param}}\hspace{-0.1in}{\basis{i}{\param}\basis{j}{\param} \pdfp
\,d\param}  = h_i^2\delta_{ij}, \eqnlabel{basisFcn}
\end{equation}
where $\delta_{ij}$ is the Kronecker delta, $h_i$ is a constant
term corresponding to $\int_{\mathcal{D}_{\param}}{\phi_i^2\pdfp\,d\param}$,
$\mathcal{D}_{\param}$ is the domain of the random variable $\param(\omega)$, and
$\pdfp$ is a probability density function for $\param$. 

Polynomial chaos theory is becoming an useful framework to study control systems with random parameters \cite{hover2006application,pctrajgen,fisher2009linear,bhattacharya2012linear,kim2012generalized,ulissi2013control,templeton2012probabilistic}.

\subsection{Application to Dynamical Systems with Random Parameters}
With respect to the dynamical system defined in \eqn{contiDyn}, the solution can be approximated by the polynomial chaos expansion as
\begin{align}
	\x(t,\param) \approx \hat{\x}(t,\param) =  \sum_{i=0}^N \x_i(t)\basis{i}{\param},
\end{align}
where the polynomial chaos coefficients $\x_i \in \real^n$. Define $\mo{\Phi}(\param)$ to be
\begin{align}
\mo{\Phi} &\equiv \mo{\Phi}(\param) := \begin{pmatrix}\basis{0}{\param} & \cdots & \basis{N}{\param}\end{pmatrix}^T, \text{ and } \\
\mo{\Phi}_n &\equiv \mo{\Phi}_n(\param) := \mo{\Phi}(\param) \otimes \I{n},
\end{align}
where $\I{n}\in\real^{n\times n}$ is identity matrix. Also define matrix $\X\in\real^{n\times(N+1)}$, with polynomial chaos coefficients $\x_i$, as
\[ \X = \begin{bmatrix} \x_0 & \cdots & \x_N \end{bmatrix}.\]

This lets us define $\hat{\x}(t,\param)$ as 
\begin{align}
\hat{\x}(t,\param) := \X(t)\mo{\Phi}(\param) \eqnlabel{compactX}.
\end{align}
Noting that $\hat{\x} \equiv \vec{\hat{\x}}$, we obtain an alternate form for \eqn{compactX},
\begin{align}
\hat{\x} \equiv  \vec{\hat{\x}} & = \vec{\X\mo{\Phi}}  = \vec{\I{n}\X\mo{\Phi}} \nonumber \\
& = (\mo{\Phi}^T\otimes \I{n})\vec{\X} = \mo{\Phi}_n^T\xpc, \eqnlabel{compactxpc}
\end{align}
where  $\xpc := \vec{\X}$, and $\vec{\cdot}$ is the vectorization operator \cite{horn2012matrix}.

Since $\hat{\x}$ from \eqn{compactxpc} is an approximation, substituting it in \eqn{simpleSys} we get equation error $\vo{e}$, which is given by 
\begin{align}
\vo{e} &:= \dot{\hat{\x}} - \A(\param)\hat{\x}\\
& =  \mo{\Phi}_n^T\xpcdot - \A(\param)\mo{\Phi}_n^T\xpc.
\end{align}
Best $\set{L}_2$ approximation is obtained by setting
\begin{align}
\inner{\vo{e}\phi_i} := \Exp{\vo{e}\phi_i}=0, \text{ for } i = 0,1,\cdots,N.
\end{align}
Upon simplification, we get a set of $n(N+1)$ \textit{deterministic} ordinary differential equations in $\xpc$, 
\begin{align}
	\xpcdot = \Exp{\mo{\Phi}\otimes\mo{\Phi}_n^T}^{-1}\Exp{\mo{\Phi}\otimes\left(\A\mo{\Phi}_n^T\right)}\xpc. \eqnlabel{pcDynamics}
\end{align}	

Using the following properties of Kronecker product
\begin{align*}
	A\otimes (B \otimes C) &= (A\otimes B) \otimes C, \\
	(A\otimes B)(C\otimes D) &= (AC)\otimes(BD),
\end{align*}
we have the following propositions.\\[5mm]
\begin{proposition} For $\mo{\Phi}_n$, and $\mo{\Phi}$ as defined earlier
\begin{equation}
	\mo{\Phi}\otimes\mo{\Phi}^T_n = (\mo{\Phi}\mo{\Phi}^T)\otimes \I{n}. \eqnlabel{id1}
\end{equation}
\end{proposition}
\begin{proof}
\begin{align*} \mo{\Phi}\otimes\mo{\Phi}^T_n &=  \mo{\Phi}\otimes(\mo{\Phi}^T\otimes \I{n}) = (\mo{\Phi}\otimes\mo{\Phi}^T)\otimes \I{n} \\
	&= \left[(\mo{\Phi} \cdot 1) \otimes (1 \cdot \mo{\Phi}^T)\right] \otimes \I{n}\\
	&= \left[(\mo{\Phi} \otimes 1)(1 \otimes \mo{\Phi}^T)\right] \otimes \I{n} \\
	&= (\mo{\Phi} \mo{\Phi}^T)\otimes \I{n}	
\end{align*}
\end{proof}

\begin{corollary}
\begin{equation}
	\Exp{\left(\mo{\Phi}\mo{\Phi}^T\right)\otimes \I{n}} = \diag\left(\Exp{\phi_i^2}\right)\otimes \I{n}. \eqnlabel{cor1}
\end{equation}
\end{corollary}
\begin{proof} Straight forward using \eqn{basisFcn}.\end{proof}

\begin{proposition}
For any matrix $\A \in \real^{n\times n}$ and $\mo{\Phi}_n$, $\mo{\Phi}$ as defined earlier
\begin{equation}
	\mo{\Phi} \otimes \left(\A \mo{\Phi}_n^T\right) = \left(\mo{\Phi}\mo{\Phi}^T\right) \otimes \A. \eqnlabel{id2}
\end{equation}
\end{proposition}
\begin{proof}
\begin{align*}
	\mo{\Phi} \otimes \left(\A \mo{\Phi}_n^T\right) & = \left(\I{N+1}\mo{\Phi}\right) \otimes \left(\A \mo{\Phi}_n^T\right) \\
	& = (\I{N+1} \otimes \A)(\mo{\Phi}\otimes\mo{\Phi}_n^T) \\
	& = (\I{N+1} \otimes \A)\left((\mo{\Phi}\mo{\Phi}^T)\otimes \I{n}\right) \\
	& = (\I{N+1} \mo{\Phi}\mo{\Phi}^T) \otimes (\A \I{n})\\
	& = (\mo{\Phi}\mo{\Phi}^T) \otimes \A.
\end{align*}	
\end{proof}
With \eqn{id1} and \eqn{id2}, we can write \eqn{pcDynamics} as
\begin{align}
	\xpcdot &= \Apc \xpc, \eqnlabel{pcDynamics:final} 
\end{align}	
where 
\begin{align}
\Apc & = \Exp{\left(\mo{\Phi}\mo{\Phi}^T\right)\otimes \I{n}}^{-1} \Exp{\left(\mo{\Phi}\mo{\Phi}^T\right) \otimes \A}, \eqnlabel{Apc}
\end{align}
and $\mo{\Phi}$, $\A$ depend on $\param$ as defined earlier. 

Let us also introduce notations $\mo{\bar{A}}:=\Exp{\mo{\Phi}_n\A\mo{\Phi}_n^T}$, $\mo{\bar{B}}:=\Exp{\mo{\Phi}_n\B\mo{\Phi}_m^T}$, $\mo{\bar{Q}}:=\Exp{\mo{\Phi}_n\mo{Q}\mo{\Phi}_n^T}$, and $\mo{\bar{R}}:=\Exp{\mo{\Phi}_m\mo{R}\mo{\Phi}_m^T}$, to compactly present the results below.

\section{Stability}
In this section, we study the exponential stability of the second mean for the dynamical system in \eqn{simpleSys}. The equilibrium solution is said to possess exponential stability of the $m^{\text{th}}$ mean if $\exists\,\delta > 0$ and constants $\alpha>0,\beta>0$ such that $\|\x_0\| < \delta$ implies $\forall t\geq t_0$
\begin{align}
\Exp{\|\x(t;\x_0,t_0)\|^m_m} \leq \beta \Exp{\|\x_0\|_m^m} e^{-\alpha (t-t_0)}.
\end{align}

It can be trivially shown that the dynamical system in \eqn{simpleSys}, with random variables $\param$, is exponentially stable in the $2^{\text{nd}}$ mean, or exponentially stable in the mean square sense (EMS-stable), if $\exists$ a Lyapunov function $V(\x):=\x^T\P\x$, with $\P=\P^T>0$, and $\alpha>0$ such that 
	\begin{align}
		\Exp{\dot{V}} \leq -\alpha \Exp{V}. \eqnlabel{stab_condition}
	\end{align}

As the sample properties are well defined, the derivation is obtained by mimicking the proof for deterministic systems. \Eqn{stab_condition} implies
\begin{align*}
	\oderiv{\Exp{V}}{t} &\leq -\alpha \Exp{V}, \\
	\Rightarrow  \Exp{\x^T\P\x} &\leq \Exp{\x_0^T\P\x_0}e^{-\alpha (t-t_0)}.	
\end{align*}
Recall for $\P=\P^T$, 
\[
\lambda_{\text{min}}(\P) \| \x \|_2^2 \leq \x^T\P\x \leq \lambda_{\text{max}}(\P) \| \x \|_2^2,
\]
\begin{align*}
\Rightarrow &\lambda_{\text{min}}(\P) \Exp{\|\x\|_2^2} \leq \lambda_{\text{max}}(\P) \Exp{\| \x_0 \|_2^2} e^{-\alpha (t-t_0)},\\
\Rightarrow & \Exp{\|\x\|_2^2} \leq \kappa(\P) \Exp{\| \x_0 \|_2^2} e^{-\alpha (t-t_0)}.
\end{align*}

EMS-stability condition in the polynomial chaos framework is presented next.

\begin{proposition}
	For any matrix $\mo{M}\in\real^{m\times n}$ and $\mo{\Phi}_n$ as defined earlier
	\begin{equation}
		\mo{M}\mo{\Phi}_n^T = \mo{\Phi}^T_m (\I{N+1}\otimes\mo{M}). \eqnlabel{prop3}
	\end{equation}
\end{proposition}
\begin{proof}
	\begin{align*}
		\mo{M}\mo{\Phi}_n^T  & = (1 \otimes \mo{M})(\mo{\Phi}^T \otimes \I{n})  \\ &= \mo{\Phi}^T \otimes \mo{M} = (\mo{\Phi}^T\I{N+1}) \otimes (\I{m}\mo{M})\\
		&= (\mo{\Phi}^T \otimes \I{m})(\I{N+1}\otimes \mo{M}) \\ & = \mo{\Phi}^T_m (\I{N+1}\otimes\mo{M}).
	\end{align*}
\end{proof}

\begin{corollary}
	\begin{equation}
		\mo{\Phi}_n\mo{M}^T = (\I{N+1}\otimes\mo{M}^T)\mo{\Phi}_m. \eqnlabel{cor2}
	\end{equation}
\end{corollary}
\begin{proof}
Take transpose of \eqn{prop3}.
\end{proof}
\vspace{5mm}
\begin{theorem}
The dynamical system in \eqn{simpleSys} is EMS-stable if $\exists$ $\P=\P^T>0$ and $\alpha > 0$ such that 
\begin{align}
	\Abar^T\set{P} + \set{P}\Abar + \alpha\Exp{\mo{\Phi}_n\mo{\Phi}_n^T}\set{P}\leq 0, \eqnlabel{pc_stab}
\end{align}
where $\set{P}:=\I{N+1}\otimes \P$.
\end{theorem}
\begin{proof}
	With $V(\x):=\x^T\P\x$, and $\P = \P^T > 0$, 
	\begin{align*}
		\dot{V}&= \dot{\x}^T\P\x + \x^T\P\dot{\x}\\
		& = \x^T\left(\A^T\P + \P\A\right)\x \\
		&= \xpc^T\left(\mo{\Phi}_n\A^T\P\mo{\Phi}_n^T + \mo{\Phi}_n\P\A\mo{\Phi}_n^T\right)\xpc.
	\end{align*}
	Using \eqn{prop3} and \eqn{cor2}, replace $\P\mo{\Phi}_n^T$ and $\mo{\Phi}_n\P$ by $\mo{\Phi}_n^T\set{P}$ and $\set{P}\mo{\Phi}_n$ respectively. This gives us
	\begin{align*}
	\dot{V} = \xpc^T\left(\mo{\Phi}_n\A^T\mo{\Phi}_n^T\set{P} +  \set{P} \mo{\Phi}_n\A^T\mo{\Phi}_N^T  \right)\xpc
	\end{align*}
	Similarly, 
	\[V = \x^T\P\x = \xpc^T\mo{\Phi}_n\P\mo{\Phi}_n^T\xpc = \xpc^T\mo{\Phi}_n\mo{\Phi}_n^T\set{P}\xpc.\]
	$\Exp{\dot{V}} \leq -\alpha \Exp{V}$ is equivalent to 
	\begin{align*}
	\Abar^T\set{P} + \set{P}\Abar + \alpha\Exp{\mo{\Phi}_n\mo{\Phi}_n^T}\set{P}\leq 0.
	\end{align*}
\end{proof}
\Eqn{pc_stab} is a convex constraint in $\P$ and $\alpha$ (pg. 11, \cite{Boydetal:1994}).

In our previous work \cite{fisher2009linear}, we presented stability conditions with Lyapunov function defined as $V(\xpc):=\xpc^T \P \xpc$, and studied stability of the deterministic system in \eqn{pcDynamics:final}. The dynamics and the Lyapunov function were both deterministic and standard Lyapunov arguments were used to analyze stability. Stability of \eqn{simpleSys} was then inferred by showing that 
\[\lim_{t\rightarrow \infty} \xpc(t)\rightarrow 0 \Rightarrow \Exp{\|\x\|^m_m} \rightarrow 0.\] Here we analyze the stability of \eqn{simpleSys} directly, with Lyapunov functions of the type $V(\x):=\x^T \P \x$. 

\section{Controller Synthesis}
In this section, condition for a state feedback controller $\K$ that achieves EMS-stability for the system in \eqn{contiDyn}, is presented. Let $V(\x):=\x^T\P\x$, where $\P = \P^T > 0$, be the Lyapunov function that certifies this. The closed-loop system with control $\u = \K\x$ is therefore
\begin{align}
	\xdot = \Bigl(\A(\param) + \B(\param)\K\Bigr)\x \eqnlabel{clp}.
\end{align}
\begin{theorem}
Closed loop system in \eqn{clp} is EMS-stable if 
\begin{align}
\set{Y}\Abar^T + \Abar\set{Y} + \set{W}^T\Bbar^T +\Bbar\set{W} + \alpha\set{Y}\Exp{\mo{\Phi}_n\mo{\Phi}_n^T}\leq 0, \eqnlabel{lmi}
\end{align}
where $\set{W} := \I{N+1}\otimes \W$, $\set{Y} := \I{N+1}\otimes \Y$, $\Y := \P^{-1}$, and $\W:=\K\Y$. 
\end{theorem}

\begin{proof}
With $V(\x):=\x^T\P\x$, where $\P = \P^T > 0$,
\begin{align*}
\dot{V} &= \x^T\left[\A^T\P +\P\A + \K^T\B^T\P + \P\B\K \right]\x, \\
& = \xpc^T \left[\mo{\Phi}_n\A^T\P\mo{\Phi}^T_n + \mo{\Phi}_n\P\A\mo{\Phi}^T_n + \right. \\ & \left.\mo{\Phi}_n\K^T\B^T\P\mo{\Phi}^T_n + \mo{\Phi}_n\P\B\K\mo{\Phi}^T_n \right]\xpc.
\end{align*}
Using \eqn{prop3} and \eqn{cor2}, we get
\begin{align*}
	\dot{V} &= \xpc^T \left[\mo{\Phi}_n\A^T\mo{\Phi}^T_n\set{P} + \set{P}\mo{\Phi}_n\A\mo{\Phi}^T_n + \set{K}^T\mo{\Phi}_n\B^T\mo{\Phi}\set{P}^T_n + \set{P}\mo{\Phi}_n\B\mo{\Phi}^T_n\set{K} \right]\xpc,
\end{align*}
where $\set{K} := \I{N+1}\otimes \K$. Therefore, $\Exp{\dot{V}} \leq -\alpha \Exp{V}$ is equivalent to 
\begin{align}
\Abar^T\set{P} + \set{P}\Abar + \set{K}^T\Bbar^T\set{P} + \set{P}\Bbar\set{K} + \alpha\Exp{\mo{\Phi}_n\mo{\Phi}_n^T}\set{P}\leq 0. \eqnlabel{bmi}
\end{align}
The above equation is nonconvex in $\set{P}$ and $\set{K}$ and can be convexified using the well known substitutions \cite{bernussou1989linear} $\Y := \P^{-1}$, and $\W:=\K\Y$. These substitutions can be written in terms of $\set{P}, \set{K}, \set{Y}$, and $\set{W}$ as
\begin{align*}
\set{W} & = \I{N+1}\otimes \W = \I{N+1}\I{N+1}\otimes \K\Y \\
& = (\I{N+1}\otimes \K)(\I{N+1}\otimes \Y) \\
& = \set{K}\set{Y}.
\end{align*}
It is also straightforward to show $\set{P} = \set{Y}^{-1}$ and $\set{K} = \set{W}\set{Y}^{-1}$. Substituting these in \eqn{bmi}, and pre-post multiplying by $\set{Y}$, we get the following convex inequality in $\Y, \W$ and $\alpha$
\begin{align*}
	&\set{Y}\Abar^T + \Abar\set{Y} + \set{W}^T\Bbar^T+\Bbar\set{W} + \alpha\set{Y}\Exp{\mo{\Phi}_n\mo{\Phi}_n^T}\leq 0. 
\end{align*}
\end{proof}
The performance parameter $\alpha$ can be maximized, along with \eqn{lmi}, as a generalized eigen-value problem (GEVP) (pg. 11, \cite{Boydetal:1994}).

\begin{theorem}
A fixed gain $\K$ that minimizes
\begin{equation}
\Exp{\int_0^\infty (\x^T\mo{Q}\x + \u^T\mo{R}\u)dt}
\eqnlabel{cost}
\end{equation}
subject to $\u = \K\x$ and dynamics given by \eqn{contiDyn} can be synthesized in the polynomial chaos framework by solving the following convex optimization problem
\[
\max \textbf{tr} \Y
\]
subject to
\[
\begin{bmatrix}
\set{Y}\mo{\bar{A}}^T + \mo{\bar{A}}\set{Y} + \set{W}^T\mo{\bar{B}}^T + \mo{\bar{B}}\set{W} & \set{Y} & \set{W}^T \\ 
\set{Y} & -\mo{\bar{Q}}{-1} & \mo{0} \\
\set{W} & \mo{0} & -\mo{\bar{R}}^{-1} 
\end{bmatrix} \leq 0.
\]
\end{theorem}

\begin{proof}
The problem is equivalent to solving
$ \min_{\P,\K} \textbf{tr} \P$, such that $\oderiv{}{t} \Exp{\x^T\P\x} \leq -\Exp{\x^T(\mo{Q} + \K^T\mo{R}\K)\x}$.
In the polynomial chaos framework, the constraint can be written as 
\begin{align*}
&\Exp{\mo{\Phi}_n(\A + \B\K)^T\P\mo{\Phi}_n^T} + \Exp{\mo{\Phi}_n\P(\A + \B\K)\mo{\Phi}_n^T} +\Exp{\mo{\Phi}_n (\mo{Q} + \K^T\mo{R}\K) \mo{\Phi}_n^T} \leq 0
\end{align*}
Using \eqn{prop3} and \eqn{cor2}, we get
\begin{align*}
& \Abar^T\set{P} + \set{P}\Abar + \set{K}^T\Bbar^T\set{P} + \set{P} \Bbar \set{K} + \Qbar + \set{K}^T\Rbar\set{K} \leq 0.
\end{align*}

Substituting $\Y = \P^{-1}$ and $\mo{W} = \K\Y$, or equivalently $\set{Y} = \set{P}^{-1}$ and $\set{W} = \set{K}\set{Y}$, and pre-post multiplying by $\set{Y}$ we get
\begin{align*}
& \set{Y}\Abar^T + \Abar\set{Y} + \set{W}^T\Bbar^T+ \Bbar\set{W} + \set{Y}\Qbar\set{Y} + \set{W}^T\Rbar\set{W} \leq 0.
\end{align*}
By redefining the cost function in terms of $\Y$ and making use of the Schur complement, we arrive at the result.
\end{proof}

In our previous work \cite{fisher2009linear}, four controller synthesis algorithms were presented, among which was the formulation $\u=\K\x$, with constant deterministic gain $\K$. The synthesis condition for this case was presented in terms of a bilinear matrix inequality, which could not be convexified. Here, we present a convex formulation for the same synthesis problem. 

Next we present Monte-Carlo based algorithms for stability analysis and controller synthesis. These will serve as benchmarks for the polynomial chaos based algorithms. Let $\{\param_s\}_{s=1}^{S} \sim p(\param)$ be samples drawn from distribution $p(\param)$, and the sample trajectories be $\x(t,\param_s)$. Define $\X_{mc}:=[\x(t,\param_1)\; \cdots \; \x(t,\param_S)]$, and $\x_{mc}:=\vec {\X_{mc}}$. Recall that for a function $F(\param)$, the expected value $\Exp{F(\param)}$ can be approximated as
\begin{align} 
& \Exp{F(\param)}:= \int_{\domain{\param}} F(\param)p(\param)d\param  \nonumber \\
& \approx \frac{1}{S} \sum_{s=1}^{S} F(\param_s), \text{ with } \param_s \sim p(\param). \eqnlabel{mcApprox}
\end{align}
Therefore, EMS-stability condition with Lyapunov function $V(\x(t,\param)) := \x^T(t,\param)\P\x(t,\param)$, $\x_s:=\x(t,\param_s)$ and $\A_s := \A(\param_s)$, can be approximated as 
\begin{align*}
 & \frac{1}{S}\sum_{s=1}^S \dot{V}(\x_s) \leq -\alpha \frac{1}{S}\sum_{s=1}^S V(\x_s), \\
\implies & \frac{1}{S}\sum_{s=1}^S \x_s^T(\A_s^T \P + \P \A_s) \x_s \leq -\alpha \frac{1}{S}\sum_{s=1}^S \x_s^T \P \x_s,\\
\implies & \frac{1}{S}\sum_{s=1}^S \x_s^T(\A_s^T \P + \P \A_s + \alpha \P )  \x_s \leq 0, \\
\implies & \frac{1}{S} \x_{mc}^T \diag(\A_s^T \P + \P \A_s + \alpha \P)_{s=1}^{S} \x_{mc} \leq 0,\\
\end{align*}
or 
\begin{equation}
\A_s^T \P + \P \A_s + \alpha \P \leq 0, \text{ for } s=1,\cdots,S.
\eqnlabel{mcStability}
\end{equation}
Similarly, a feasible gain $\K$ can be obtained by satisfying 
\begin{align}
\Y \A_s^T + \A_s \Y + \W^T \B_s^T + \B_s \W  + \alpha \Y \leq 0,
\eqnlabel{mc:feasible}
\end{align}
for $s=1,\cdots,S$, with $\B_s:=\B(\param_s)$. Finally, the optimal gain for cost function given by \eqn{cost}, can be obtained by solving the following optimization problem
\[
\max \textbf{tr} \Y
\]
subject to 
\begin{equation}
\small
\begin{bmatrix}
\Y\A_s^T + \A_s\Y + \W^T \B_s^T + \B_s\W & \Y & \W^T \\ 
\Y & -\mo{Q}^{-1} & \mo{0} \\
\W &  \mo{0} & -\mo{R}^{-1} 
\end{bmatrix} \leq 0,
\eqnlabel{mc:optimal}
\end{equation}
for $s=1,\cdots,S$. Conditions in \eqn{mc:feasible} and \eqn{mc:optimal} are derived by approximating the expectation integral using \eqn{mcApprox} and following the argument presented in the derivation of \eqn{mcStability}.

\section{Example}
The plant considered here is an F-16 aircraft at high angle of attack \cite{fenwu:f16}, with states $x=[V\;\;\alpha\;\;q\;\;\theta\;\;T]^T$, where $V$ is the velocity, $\alpha$ the angle of attack, $q$ the pitch rate, $\theta$ is the pitch angle, and $T$ is the thrust.  The controls, $u=[\delta_e\;\delta_{th}]^T$, are the elevator deflection $\delta_e$, and the throttle $\delta_{th}$. This linear model is obtained about $\alpha=35^o$, where the moderately high angle of attack causes inaccurate modeling of aerodynamic coefficients and thus results in parametric uncertainty. The $A(\param)$ and $B$ matrices are given by

\begin{equation} \eqnlabel{f16}
A(\param)=\left[\begin{array}{ccccc}0.1658&-13.1013& -7.2748(1+0.4\Delta) &-32.1739&0.2780\\
0.0018&-0.1301&0.9276(1+0.4\Delta)&0&-0.0012\\
0&-0.6436&-0.4763&0&0\\
0&0&1&0&0\\ 0&0&0&0&-1
\end{array}\right], \;
B=\left[\begin{array}{cc}0&-0.0706\\0&-0.0004\\0&-0.0157\\0&0\\64.94&0
\end{array}\right].
\end{equation}

Similar to the model in \cite{fenwu:f16}, the terms in parenthesis in the $A(\param)$ matrix are assumed to be uncertain and are functions of a \textit{single} random variable $\param$, uniformly distributed over $[-1,1]$. This uncertainty is due to the uncertainty in the damping term $C_{xq}$, which is difficult to model in high angle of attack. The uncertainty is assumed to be distributed uniformly by $\pm 40\%$ about the nominal values $-7.2748$ and $0.9276$ respectively. \Fig{result} shows the variation of $\|\P^\ast\|$ obtained by solving the optimization problem outlined above for both Monte-Carlo and polynomial chaos based formulation. The objective function is defined by $\mo{Q}:=\diag([1E-1, 1E2, 1E-2, 1E0, 1E-3])$, and $\mo{R}:= \diag([1E2, 5E-3])$. In the Monte-Carlo approach, the problem is solved with sample sizes $5, 10, 50, 100, 200, 500, 1000, 1500,$ and $2000$ respectively. To capture the confidence in the solution for each sample size, the optimization problem is solved 100 times with new samples. \Fig{result} shows the mean $\|\P^\ast\|$ and its standard deviation. We can observe that the mean $\|\P^\ast\|$ converges and the confidence in the solution also improves as we increase the sample size. \Fig{result} also shows the convergence of $\|\P^\ast\|$ with increasing order of polynomial chaos approximation, solved with approximation order $1,\cdots,20$. Since the polynomial chaos approach is a deterministic framework, there is no variability in the answer for a given approximation order. The $x$-axis in \fig{result} is the actual number variables in the optimization problem that is solved and is obtained from \texttt{CVX}\cite{grant2008cvx}/\texttt{SDPT3}\cite{toh1999sdpt3}. This is proportional to the number of samples in Monte-Carlo approach and the order of approximation in polynomial chaos approach. Comparing the two plots we see that the trend with polynomial chaos is consistent with mean trend from Monte-Carlo. However, to get a high confidence solution, the Monte-Carlo approach requires to solve a problem with an order of magnitude higher in size. With polynomial chaos approach we can get close to the converged solution with fairly low order approximation, and thus offers a computational advantage for the class of problems considered here. Although, \fig{result} shows this trend for the specific system, we can expect to see this advantage for general systems as well. This is because, in general, polynomial chaos is computationally more efficient than Monte-Carlo for characterizing uncertainty in dynamical systems.

\begin{figure}[h!]
\includegraphics[width=\textwidth]{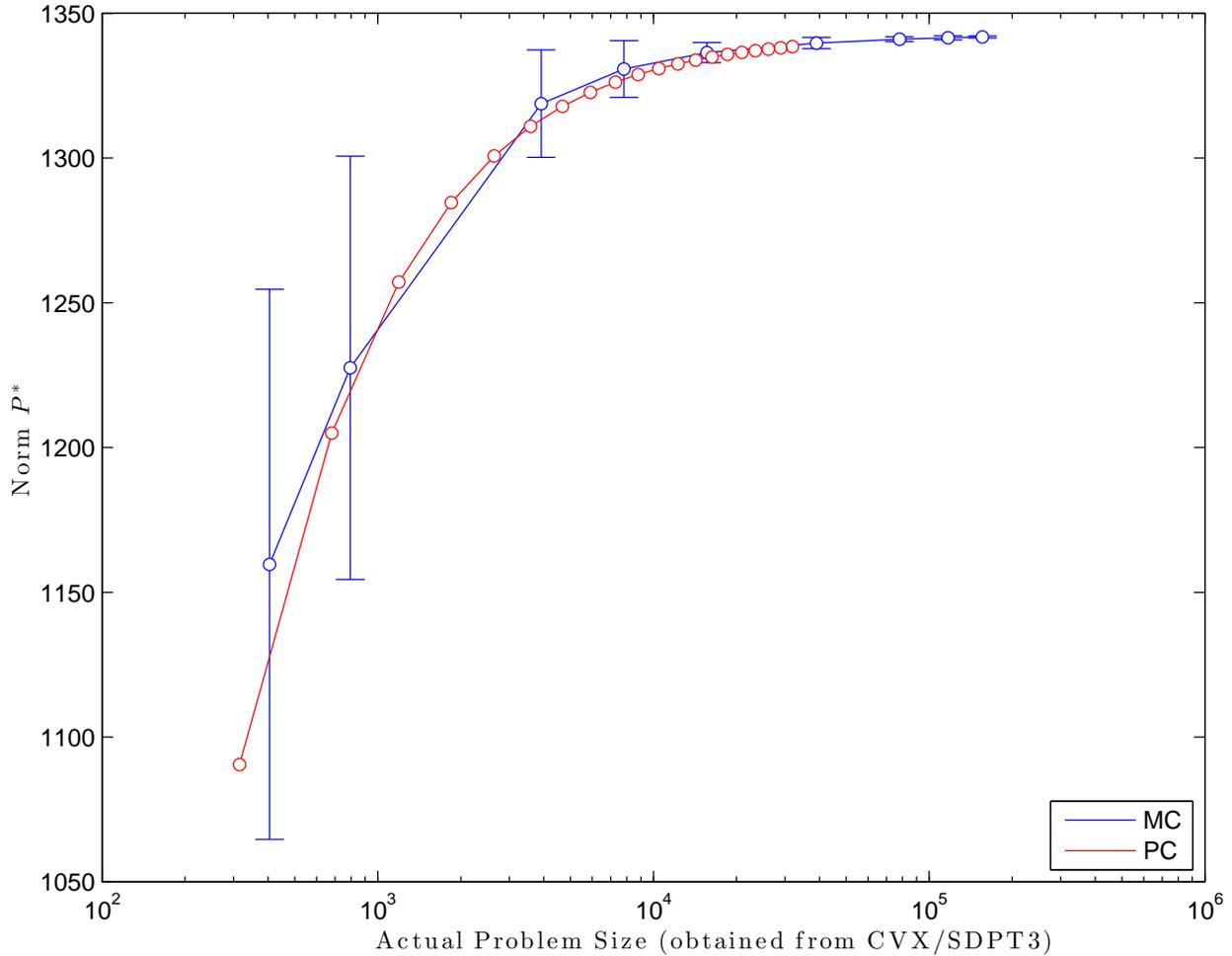}
\caption{Variation of $\|\P^\ast\|$ with optimization problem size, for Monte-Carlo (mean and standard deviation) and polynomial chaos based controller synthesis algorithms. Problem size is proportional to the number of samples in Monte-Carlo approach and the order of approximation in polynomial chaos approach. }
\figlabel{result}
\end{figure}

\section{Summary}
In this paper we presented new convex conditions for EMS-stability of linear dynamical systems with probabilistic uncertainty in system parameters. These results were obtained using the polynomial chaos framework and are similar to the well known results for deterministic systems. We applied this framework to design EMS-stabilizing controllers for an F-16 aircraft model and demonstrated the computational advantage of polynomial chaos based approach over Monte-Carlo based approach for analyzing dynamical systems with random parameters.
\bibliographystyle{IEEEtran}
\bibliography{/Users/raktim/Library/texmf/tex/latex/raktim}

\end{document}